\documentclass[preprint,12pt,english]{article}
\usepackage{color,graphicx}
\usepackage{textcomp}
\usepackage{amsmath}
\usepackage{amssymb}
\usepackage{amsthm}
\usepackage{latexsym}
\newcommand{\junk}[1]{}

\newtheorem{theorem}{Theorem}
\newtheorem{corollary}{Corollary}
\newtheorem{lemma}{Lemma}

\newtheorem{remark}{Remark}

\newtheorem{definition}{Definition}

\newtheorem{proposition}{Proposition}

\def\OR{\vee}
\def\AND{\wedge}
\def\bigdoublewedge{\bigwedge}
\def\bigdoublevee{\bigvee}

\begin{document}

\pagestyle{empty}

\title{Proof Complexity and the Kneser-Lov\'asz Theorem (I)} 

\author{
Gabriel Istrate, 
Adrian Cr\~aciun\footnote {Dept. of Computer Science, West University of Timi\c{s}oara and e-Austria Research Institute, Bd. V. P\^{a}rvan 4, cam. 045
B, Timi\c{s}oara, RO-300223, Romania. Corresponding author's email: {\tt gabrielistrate@acm.org }}}

\maketitle

\begin{abstract} 
We investigate the proof complexity of a class of propositional formulas 
expressing a combinatorial principle known as {\em the Kneser-Lov\'{a}sz  Theorem}. 
This is a family of propositional tautologies, indexed by an nonnegative integer
parameter $k$ that generalizes the Pigeonhole Principle
(obtained for $k=1$). 

We show, for all fixed $k$, 
$2^{\Omega(n)}$ lower bounds on resolution complexity and exponential lower bounds for bounded depth Frege proofs. These results 
hold even for the more restricted class of formulas encoding Schrijver's strenghtening of the Kneser-Lov\'{a}sz Theorem. On the other hand for the cases $k=2,3$ (for which combinatorial proofs of the Kneser-Lov\'{a}sz Theorem are known) we give polynomial size Frege ($k=2$), respectively extended Frege ($k=3$) proofs. The paper concludes with a brief announcement of the results (presented in subsequent work) on the complexity of the general case of the Kneser-Lov\'{a}sz theorem.

\end{abstract}

\section{Introduction}

\label{intro}  One of the most interesting approaches in discrete mathematics is the use of 
 topological methods to prove results having a purely combinatorial nature. The approach started 
 with Lov\'asz's proof \cite{lovasz-kneser} of a combinatorial 
 statement raised as an open problem by Kneser in 1955 (see \cite{longueville-ems} for a historical account). A significant amount of 
 work has resulted from this conjecture (to get a feel for the advances  consult \cite{matousek-book,kozlov-book}).

Methods from topological combinatorics raise interesting challenges from a complexity-theoretic  point of view: they are non-constructive, often based on principles that appear to lack polynomial time algorithms (e.g. Sperner's Lemma and the Borsuk-Ulam Theorem \cite{papadimitriou1994complexity-parity}). The concepts involved (simplicial complexes, chains, chain maps) seem to  require intrinsically exponential size representations. 

In this paper we raise the possibility of using statements from topological combinatorics as a source of interesting candidates for proof complexity. In particular we view the Kneser-Lov\'{a}sz theorem as a statement on the unsatisfiability of a certain class of propositional formulas, and investigate the complexity of proving their unsatisfiability.

We were initially motivated by the problem of separating the Frege
and extended Frege proof systems. Various candidate formulas have
been proposed (see \cite{bonet-buss-pitassi} for a discussion). It was natural to wonder whether the non-elementary nature of mathematical proofs of Kneser's theorem translates into hardness and separation results in propositional complexity. We no longer believe that this problem provides such examples. Yet gauging its precise 
complexity  is still, we feel, interesting. 

A slightly different perspective on this problem is the following: Matou\v{s}ek obtained \cite{matousek-combinatorial-kneser} a "purely combinatorial" proof of the Kneser-Lov\'asz theorem, a proof that does not explicitly mention any topological concept. While combinatorial, Matou\v{s}ek's proof is nonconstructive: the approach in 
\cite{matousek-combinatorial-kneser} "hides" in purely combinatorial terms the application of the so-called {\em Octahedral Tucker Lemma}, a discrete variant of the Borsuk-Ulam theorem. Searching for the object guaranteed to exist by this principle, though "constructive" in theory \cite{freund1981constructive} is likely to be intractable, as the associated search problem for the 2-d Tucker lemma\footnote{As kindly pointed to us by professor P\'alv\H{o}lgyi this is also likely but not explicitly proved in \cite{palvolgyi20092d} for the octahedral Tucker lemma.} is complete for the class PPAD \cite{palvolgyi20092d}. 

Thus another perspective on the main question we are interested in is {\bf under what circumstances do cases of the Kneser-Lov\'asz theorem  have combinatorial proofs of polynomial size}. This depends, of course, on the proof system considered, making the question fit the "bounded reverse mathematics" program of Cook and Nguyen \cite{cook2010logical}. A natural boundary seems to be the class of Frege proofs: for $k=1$ the Kneser-Lov\'asz theorem is equivalent to the pigeonhole principle (PHP) that has polynomial size $TC^{0}$-Frege proofs, but exponential lower bounds in  resolution \cite{buss-frege-php} and bounded depth Frege. On the other hand {\bf obtaining a similar upper bound} for the general case would be quite significant, as {\bf it would seem to require completely bypassing the techniques from Algebraic Topology} starting instead from radically different principles.

Our contributions (and the outline of the paper) can be summarized as follows: In Section~\ref{sec:lb} we give a reduction between $Kneser_{k,n}$ and $Kneser_{k+1,n}$ for arbitrary $k\geq 1$. As an application we infer that existing lower bounds for PHP apply to formulas $Kneser_{k,n}$ for any fixed value of $k$. In Section~\ref{sec:23} we investigate cases $k=2,3$ (when the Kneser-Lov\'asz theorem has combinatorial proofs). We give Frege proofs (for $k=2$) and extended Frege proofs (for $k=3$), both having polynomial size. 

As usual in the case of bounded reverse mathematics, our positive results could have been made uniform by stating them (more carefully) as expressibility results in certain logics: for instance our result for the case $k=2$ of the Kneser-Lov\'asz theorem could be strengthened to an expressibility result in logical theory $VNC^1$ \cite{cook2010logical}. We will not pursue this approach in the paper, deferring it to the journal version.

\section{Preliminaries} 

Throughout this paper $k$ will be a fixed constant greater or equal to 1. Given a set of integers $A$, we will denote by ${{A}\choose {k}}$ the set of cardinality $k$ subsets of set $A$. We will write $|A|$ instead of $A$ in the previous definition in case $A=\{1,2,\ldots, n\}$ for some $n\geq 1$. $A\subseteq [n]$ will be called {\em stable} if for no $1\leq i\leq n$ both $i$ and $i+1 \mbox{ (mod n) }$ are in $A$. Also denote by $A_{\leq k}$ (called "firsts of A") the set of smallest (at most) $k$ elements of $A$.

The Kneser-Lov\'{a}sz theorem is formally stated as follows:

\begin{proposition} 
Given $n\geq 2k\geq 1$ and a function $c:{{n}\choose {k}}\rightarrow [n-2k+1]$ there exist two disjoint sets $A,B$ and a color 
$1\leq l\leq n-2k+1$ with $c(A)=c(B)=l$. 
\label{kl}
\end{proposition}

An even stronger form was proved by Schrijver \cite{schrijver1978vertex}: Proposition~\ref{kl} is true if we limit the domain of $c$ to all stable subsets\footnote{we will denote this collection of sets by ${{n}\choose {k}}_{stab}$} of $[n]$ of cardinality $k$:

\begin{proposition} 
Given $n\geq 2k\geq 1$ and a function $c:{{n}\choose {k}}_{stab}\rightarrow [n-2k+1]$ there exist two disjoint sets $A,B$ and a color 
$1\leq l\leq n-2k+1$ with $c(A)=c(B)=l$. 
\label{shcrij}
\end{proposition} 

The Kneser-Lov\'{a}sz Theorem can be seen as a statement about the
chromatic number of a particular graph: define the graph $KG_{n,k}$
to consist of the subsets of cardinality $k$ of $[n]$, connected by an
edge when the corresponding sets are disjoint (Figure~\ref{kg-graph}). Then the
Kneser-Lov\'{a}sz Theorem is equivalent to $\chi(KG_{n,k})\geq
n-2k+1$ (in fact $\chi(KG_{n,k})= n-2k+2$, since the upper bound is
easy \cite{matousek-book}).

\begin{figure}[ht]
\begin{center}
\includegraphics[height=3cm]{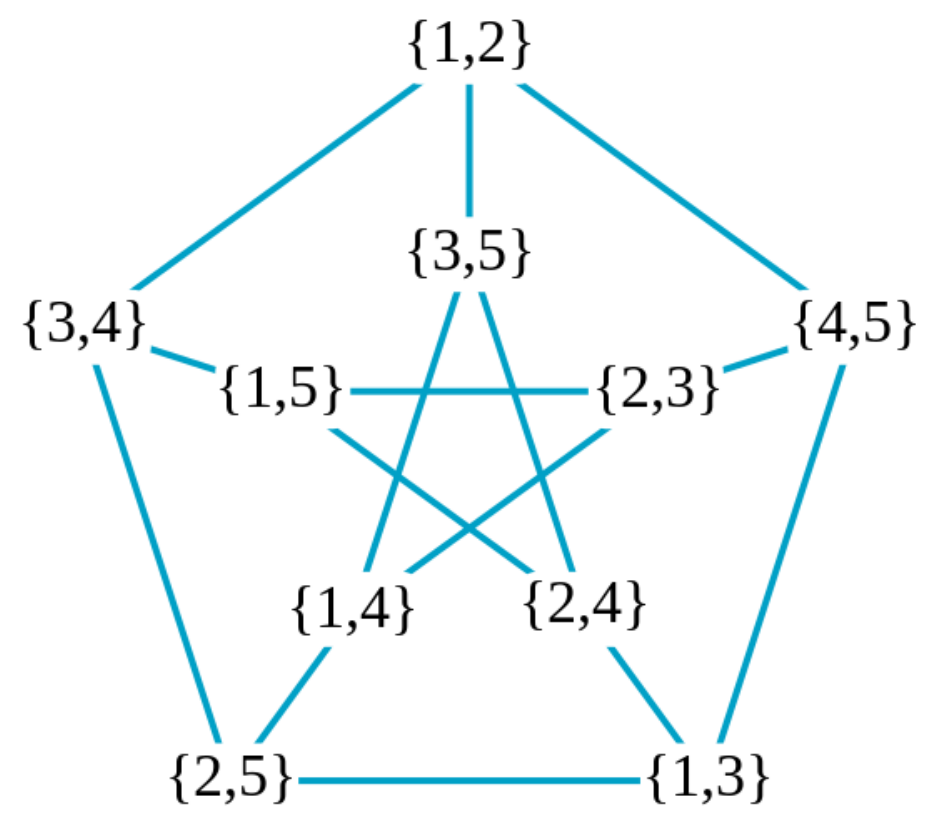}
\end{center}
\caption{Kneser graph $KG_{5,2}$ a.k.a. the Petersen graph. The Kneser-Lov\'asz Theorem states that the chromatic number of this graph is 3. Schrijver's Theorem claims that a similar result holds for the interior star only}
\label{kg-graph}
\end{figure}

We assume familiarity with the basics of proof complexity, as presented for instance in \cite{krajicek-book}, in particular with resolution complexity (the size measure will be denoted by $res$), Frege, extended Frege (EF) proofs and the concepts and results in \cite{buss-frege-php}. We will state our positive results using the sequent calculus system LK \cite{krajicek-book}, a system  $p$-equivalent to Frege proofs.

%Given $F$ an unsatisfiable formula we will denote by 
%$res(F)$ the {\em length} (defined as number of clauses) of the smallest resolution refutation of $F$. 

\begin{definition}
Let {\bf $PHP_{n}^{m}$} be the formula 
$
 \bigdoublewedge\limits_{i=1}^{m} (\bigdoublevee\limits_{l=1}^{n} X_{i,l})\vdash
 \bigdoublevee\limits_{i\neq j}\big[\bigdoublevee\limits_{l=1}^{n} (X_{i,l}\wedge
 X_{j,l})\big].$
\end{definition}

$PHP_{n}^{n+1}$ has polynomial time Frege proofs \cite{buss-frege-php}. An important ingredient of the proof is the representation of natural numbers as sequences of bits, with every bit being expressed as the truth value of a certain formula. We will use a similar strategy. In particular quantities such as ${{n}\choose {2}}$ will refer to the logical encoding of the binary expansion of integer $\frac{n\cdot (n-1)}{2}$. We will further identify  statements such as "$A=B$" or "$A\leq B$" with the logical formulas expressing them. The approach of Buss uses {\em counting}, defining a set of families of formulas $Count_{n}$, such that  
$Count_{n}(Y_{1},\ldots, Y_{n})$ yields the binary encoding of the number of variables $Y_{1},\ldots, Y_{n}$ that are TRUE. We will often drop the index $n$ from notation if its value is self-evident. We will further need several 
simple intentional properties of function $Count$ with respect to combinatorics. Formal arguments are deferred to the journal version.   

\begin{lemma} Let $n\leq m$. and let $X_{1},\ldots X_{n}, Y_{1},\ldots Y_{m}$ be logical variables
. In $LK$ one can give polynomial-size proofs of the following facts:
\begin{enumerate}
\item $X_{1}\AND X_{2}\AND \ldots X_{n}\vdash Count_{n}[X_{1},X_{2},\dots, X_{n}]=n$. 
\item Let $X_{1},X_{2},\ldots X_{n}$ be logical variables. Then 
\[
\vdash Count_{{{n}\choose {2}}}[X_{1}\AND X_{2}, \ldots, X_{i}\AND X_{j},\ldots ,X_{n-1}\AND X_{n}]={{Count_{n}[X_{1},X_{2},\dots, X_{n}]}\choose {2}}
\]
\item 
Let $X_{1},X_{2},\ldots X_{n}$ be logical variables. Then
\[
\vdash Count_{n^2}[X_{i}\AND \delta_{\{i\neq j\}}]=Count_{n}[X_{1},X_{2},\dots, X_{n}]\cdot (n-1).
\]
\item 
\[
X_{1}\leq Y_{1}, \ldots, X_{n}\leq Y_{n}\vdash Count_{n}[(X_{i})]\leq Count_{m}[(Y_{j})].
\]
\end{enumerate} 
\label{frege-count}
\end{lemma} 

Finally a {\em variable substitution} in a formula will refer in this paper to substituting every variable by some other variable (not necessarily in a 1-1 manner). 

\iffalse
We will employ the following definition from \cite{buss2006polynomial} in order to derive lower bounds on constant-depth Frege proofs: 

\begin{definition} 
Let $Q$ and $T$ be families of propositional formulas. Let $F + T$
denote a Frege system augmented to include all substitution instances of
formulas from $T$ . Then, we say $Q\preceq_{cdF} T$ holds provided that the formulas
from $Q$ have polynomial size, constant depth proofs in the proof system
$F +T$.
We write $Q \equiv_{cdF} T$ to mean that both $Q\preceq_{cdF} T$ and $T\preceq_{cdF} Q$.
\label{cd-red}
\end{definition} 
\fi

\subsection{Propositional formulation of the Kneser-Lov\'{a}sz Theorem}

We define a variable $X_{A,l}$ for every set $A\in {{n}\choose {k}}$ of
cardinality $k$, and partition class $P_{l}:=c^{-1}(\{l\})$. $X_{A,l}$ is intended
to be TRUE iff $A\in P_{l}$ and zero otherwise.

\begin{definition}

Denote by 

\begin{itemize}
\item {\bf $Ant_{k,n}$} the formula $\bigdoublewedge\limits_{A\in {{n}\choose {k}}} (\bigdoublevee\limits_{l=1}^{n-2k+1} X_{A,l})$. 

\item {\bf $Cons_{k,n}$} the formula $\bigdoublevee\limits_{\stackrel{A,B\in {{n}\choose {k}}}{A\cap B=\emptyset}}\big(\bigdoublevee\limits_{l=1}^{n-2k+1} (X_{A,l}\wedge
 X_{B,l})\big)$. 
\item {\bf $Onto_{k,n}$} the formula $\bigdoublevee\limits_{A\in {{n}\choose {k}}}\big(\bigdoublevee\limits_{\stackrel{l,s =1}{l\neq s}}^{n-2k+1} (\overline{X_{A,l}}\OR 
 \overline{X_{A,s}})\big)$
\item Finally, denote by {\bf $Kneser_{k,n}$} the formula 
$[Ant_{k,n}\vdash Cons_{k,n}]$. 
$Kneser_{k,n}$ is (by \cite{lovasz-kneser}) a tautology  with $(n-2k+1)\cdot {{{n} \choose {k}}}$
variables.
\item We will also encode the {\em onto version} of the Kneser-Lov\'asz Theorem. Indeed, denote by {\bf $Kneser^{onto}_{k,n}$} the formula 
$[Ant_{k,n}\AND Onto_{k,n}\vdash Cons_{k,n}].$
\end{itemize} 
\end{definition}

Note that formula $Kneser_{1,n}$ is essentially the Pigeonhole principle $PHP_{n-1}^{n}$. 

\section{\label{sec:lb}Lower bounds: Resolution Complexity and bounded-depth Frege proofs}

The following result shows that many lower bounds on the complexity of the pigeonhole principle apply directly 
to any family $(Kneser_{k,n})_{n}$:

\begin{theorem} For all $k\geq 1,n\geq 3$ there exists a variable substitution 
$\Phi_{k}$, \\ $\Phi_{k}:Var(Kneser_{k+1,n})$ $\longrightarrow Var(Kneser_{k,n-2})$
 such that $\Phi_{k}(Kneser_{k+1,n})$ is a formula consisting precisely of the clauses of 
$Kneser_{k,n-2}$ (perhaps repeated and in a different order). 
\label{subst-kneser} 
\end{theorem} 
\begin{proof} 

For simplicity we will use different notations for the sets of variables of the two formulas: we assume that 
$Var(Kneser_{k+1,n})=\{X_{A,i}\}$ and $Var(Kneser_{k,n-2})=\{Y_{A,i}\}$, with obvious (different) ranges for $i$ and $A$. 

Let $A\in {{n}\choose {k+1}}$. For $1\leq i \leq n-2(k+1)+1=n-2k-1$ define $\Phi_{k}(X_{A,i})$ by: 

\begin{itemize} 
\item{\bf Case 1: $A_{\leq k}\subseteq [n-2]$:} Define
\begin{equation} 
\Phi_{k}(X_{A,i})=Y_{A_{\leq k},i}
\end{equation} 

\item{\bf Case 2:$A_{\leq k}\not \subseteq [n-2]$:} \\

In this case necessarily both $n-1$ and $n$ are members of $A$. 

Let $A=P\cup \{n-1,n\}$, $|P|=k-1$. Let $\lambda=max\{j : j\leq n-2, j\not \in P \}$. Define 

\begin{equation}
\Phi_{k}(X_{A,i})=Y_{P\cup \{\lambda\},i}.
\end{equation}
\end{itemize} 

Formula $Kneser_{k,n-2}$ has clauses of two types
\begin{itemize} 
\item{(a).} Clauses of type $Y_{A,1}\vee Y_{A,2}\vee \ldots Y_{A,n-2k-1}$, with $A\in {{n-2}\choose {k}}$. 
\item{(b).} Clauses of type $\overline{Y_{A,i}}\vee \overline{Y_{B,i}}$ with $1\leq i\leq n-2k-1$, $A,B\subseteq {{n-2}\choose {k}}, A\cap B=\emptyset$. 
\end{itemize} 

As $\Phi_{k}$ preserves the second index, every clause of type (a) of $Kneser_{k+1,n}$ maps via $\Phi_{k}$ to a clause of type (a) of $Kneser_{k,n-2}$. On the other hand every clause of type (a) is the image through $\Phi_{k}$ of some clause of $Kneser_{k,n+1}$, for instance of clause $X_{C,1}\vee X_{C,2}\vee \ldots \vee X_{C,n-2k-1}$, where $C=A\cup \{n-1\}$. 

As for clause $\overline{X_{A,i}}\vee \overline{X_{B,i}}$ of type (b), again we use the fact that $\Phi_{k}$ preserves the second index, and prove that the substituted variables correspond to disjoint subsets: 

\begin{itemize} 
\item{\bf Case I:} $A,B$ both fall in Case 1. of the definition of $\Phi_{k}$. 

Denote for simplicity $C=A_{\leq k},D=B_{\leq k}$, hence $\Phi_{k}(\overline{X_{A,i}}\vee \overline{X_{B,i}})=\overline{Y_{C,i}}\vee \overline{Y_{D,i}}$).

It follows that $C,D$ are disjoint (as $A\cap B=\emptyset$ and $C\subseteq A$, $D\subseteq B$). Note that the converse is also true: every clause $\overline{Y_{C,i}}\vee \overline{Y_{D,i}}$ is the image of clause 
$\overline{X_{A,i}}\vee \overline{X_{B,i}}$, with $A=C\cup \{n-1\}$, $B=D\cup \{n\}$. 

\item{\bf Case II:} One of the sets, say $A$, falls under Case 2, the other one, $B$,  falls under Case 1 (note that  $A$ and $B$ cannot both fall under Case 2, as they would  both contain $n-1, n$ and they would no longer be disjoint). In this case $C=P\cup \{\lambda\}, D=B_{\leq k}.$ As $\{\lambda+1,\ldots, n\} \subset A$ and $A\cap B=\emptyset$, $
\lambda+1,\ldots, n\not \in B.$ Therefore, even though it might be possible that $\lambda\in B$, certainly $\lambda\not \in B_{\leq k}$ (since there are no elements in $B$ larger than $\lambda$). Thus $C\cap D=(P\cup \{\lambda\})\cap B_{\leq k}\subseteq A\cap B=\emptyset$. 
\end{itemize} 
\end{proof}\qed

The previous result can be applied $k$ times to show the following two lower bounds:  

\begin{theorem}
For any fixed $k\geq 1$ we have $res(Kneser_{n,k})=2^{\Omega(n)}$ (where the constant might depend on $k$).
\label{res-lb}
\end{theorem} 
\begin{proof} 
The result follows from the following simple 

\begin{lemma} 
Let $\Phi$ be a propositional formula let $X\stackrel{\phi}{\rightarrow} Y$ be a variable substitution and let 
$\Xi = \Phi[X\stackrel{\phi}{\rightarrow} Y]$ be the resulting formula. Assume that $P=C_{1},C_{2},\ldots, C_{r}$ is a resolution refutation of $\Phi$ and let $\phi(P)=\phi(C_{1}),\phi(C_{2}),\ldots, \phi(C_{r})$. Then $\phi(P)$ is a resolution refutation of $\Xi$. Consequently $res(\Xi)\leq res(\Phi)$. 
\label{subst-resolution} 
\end{lemma} 
\begin{proof} 

Similar, more powerful (less trivial) results of this type were explicitly stated, e.g. in \cite{ben2011understanding}. 
\end{proof}\qed

%\begin{remark} 
%Lemma~\ref{subst-resolution} was stated for resolution. However an analog of it holds (and can be proved similarly) for Frege proofs. We will need this %observation later. 
%\label{subst-frege} 
%\end{remark} 
\end{proof}\qed

Similarly 
\begin{theorem} 
For any fixed $k\geq 1$ and arbitrary $d\geq 1$ there exists $\epsilon_{d}>0$ such that the family $(Kneser_{n,k})$ has 
$\Omega(2^{n^{\epsilon_{d} }})$ depth-$d$ Frege proofs
\label{bdfrege-lb}
\end{theorem} 
\begin{proof} 

We employ the  the corresponding bound for $PHP^{n}_{n-1}(=Kneser_{1,n})$ \cite{krajicek-pudlak-woods}.
\end{proof}\qed 

\iffalse
Note that our reduction has corresponding implications for the {\em parameterized resolution complexity} \cite{dantchev2011parameterized} of formulas $Kneser_{k,n}$ (with the so-called {\em naive embedding}): corresponding lower bounds for $PHP_{n}^{n+1}$ \cite{dantchev2011parameterized} extend to  $Kneser_{k,n}$ for any fixed $k$. The (rather easy) details will be provided in the full version. 
\fi 
\subsection{Extension: lower bounds on the proof complexity of Schrijver's theorem} 

We can prove (stronger) bounds similar to those of Theorems~\ref{res-lb} and~\ref{bdfrege-lb} for Schrijver's formulas by noting that the following variant of Theorem~\ref{subst-kneser} holds: 

\begin{theorem} 
For every $k\geq 1, n\geq 3$ there exists a variable substitution $\Phi_{k}$, $\Phi_{k}:Var(Sch_{k+1,n})$ $\longrightarrow Var(Sch_{k,n-2})$
 such that $\Phi_{k}(Sch_{k+1,n})$ is a formula consisting precisely of the clauses of 
$Sch_{k,n-2}$ (perhaps repeated and in a different order). 
\label{subst-tucker}
\end{theorem} 

\begin{proof} 
Substitution $\Phi_{k}$ is exactly the same as the one in the proof of Theorem~\ref{subst-kneser}. In this case we need to further argue three things: 
\begin{itemize} 
\item[(1)] If $\Phi_{k}$ maps $X_{A,i}$ onto  $Y_{C,i}$ and $A$ is stable then so is $C$. 
\item[(2)] Every clause $Y_{C,1}\vee Y_{C,2}\vee \ldots \vee Y_{C,n-2k-1}$ of $Sch_{k,n-2}$ is the image of a clause $X_{A,1}\vee X_{A,2}\vee \ldots \vee X_{A,n-2k-1}$ with $A$ stable. 
\item[(3)] Every clause $\overline{Y_{C,i}}\vee \overline{Y_{D,i}}$ of $Sch_{k,n-2}$ is the image of a clause $\overline{X_{A,i}}\vee \overline{X_{B,i}}$ 
with $A,B$ disjoint and {\em stable}.
\end{itemize} 

\begin{itemize} 
\item[(1)] If $A_{\leq k}\subseteq [n-2]$ then $C=A_{\leq k}$ satisfies the stability condition everywhere except perhaps at elements 1 and n-2. But if $1\in C\subseteq A$ then $n\not \in A$ (as $A$ is stable). Similarly $n-1\not \in A$. This contradicts the fact that $A$ must contain one of $n-1,n$. 

On the other hand it is not possible that $A$ falls under Case 2, as it would have to contain successive elements $n-1,n$. 
\item[(2)] Since $C$ is stable, one of  $1,n-2$ is {\bf not} in $C$. Define $A$ to consist of $C$ together with the unique element in $n-1,n$ not forbidden by stability. 
\item[(3)] Similarly to (2): given disjoint stable sets $C$,$D$ in $[n-2]$ obtain $A$ and $B$ by adding the elements $n-1,n$ to $C,D$, one to each set, respecting the stability condition. This is possible as $C$ and $D$ are disjoint. For instance, if $n-2\in C$ then $1\not \in C$, and we distribute $n$ in $C$ and $n-1$ in $D$. 
\end{itemize} 
\end{proof}\qed

\section{\label{sec:23}The cases $k=2$ and $k=3$ of the Kneser-Lov\'asz Theorem}

Unlike the general case, for $k\in \{2,3\}$ Kneser's conjecture has
combinatorial proofs \cite{stahl-jcb},\cite{garey-johnson-coloring}. This facts motivates the following
theorem, similar to the one proved in \cite{buss-frege-php}  for the
Pigeonhole Principle:

\begin{theorem} The following are true: 
\begin{itemize}
\item{(a)} The class of formulas $Kneser^{onto}_{2,n}$ has polynomial size Frege proofs.
\item{(b)} The class of formulas $Kneser^{onto}_{3,n}$ has polynomial size {\em extended} Frege proofs. 
\end{itemize} 
\label{ub}
\end{theorem}

\begin{proof}

Informally, the basis for the combinatorial proofs in \cite{garey-johnson-coloring},
\cite{stahl-jcb} of cases $k\in \{2,3\}$ is the following claim, only valid for these values of $k$:  any partition of ${{n}\choose {k}}$ into classes $P_{1}, P_{2}, \ldots ,P_{n-2k+1}$ contains at least one class $P_{j}$ such that either $\bigcap_{A\in P_{j}} A \neq \emptyset. \mbox{ or }
 A\cap B=\emptyset \mbox{ for some }A,B\in P_{j}$.  

This claim could be used  as the basis for the propositional simulation of the proofs from  
\cite{stahl-jcb} and \cite{garey-johnson-coloring}, respectively. 
This strategy only leads to {\em extended Frege}, rather than Frege proofs for $Kneser_{k,n}$. The reason is that we eliminate one element from $\{1,\ldots, n\}$ and one class from the partition. Similar to the case of PHP in \cite{buss-frege-php}, doing so involves renaming, leading to extended Frege proofs. 

For $k=2$ we will bypass the problem above by giving a stronger, counting-based proof of $Kneser_{2,n}$.
We will then explain why a similar strategy apparently does not work for $k=3$ as well.
In both situations, $k\in \{2,3\}$ below we first present the mathematical argument, then discuss how to formalize it in 
(extended) Frege. 

%We will show that these formulas have polynomial size {\em cutting plane} proofs. This is similar to the case of $PHP_{n}$, for which such a result was proved by Cook, Coulard and Turan \cite{cook1987complexity}. Later on Goerdt \cite{goerdt1991cutting} has shown that Frege proofs polynomially simulate cutting plane proofs. We will employ this result too. 

\subsection{\bf Case $k=2$ }

%\begin{lemma} 
%Let $1\leq l\leq n-3$ Then 
%\[
%\sum_{i=1}^{l} |c^{-1}(i)|\leq (n-1)+(n-2)+\ldots + (n-l).
%\]
%\label{count:2}
%\end{lemma}  

\subsubsection*{Mathematical (semantic) proof.} 

The result follows from the following sequence of claims: 

\begin{lemma} Given any (n-3)-coloring $c$ of ${{n}\choose {2}}$ and color $1\leq l \leq n-3$, at least one of the following alternatives is true: 
\begin{enumerate}
\item there exist two disjoint sets $D,E\in c^{-1}(l)$.
\item  $|c^{-1}(l)|\leq 3$.
\item there exists $x\in [n]$, $x\in \bigcap\limits_{A\in c^{-1}(l)} A$.
\end{enumerate} 
\label{claim1}
\end{lemma} 
 
\begin{proof} 
Assume that $D=\{a,b\}\in c^{-1}(l)$ and there is a set $E\in c^{-1}(l)$, $a\not\in E$, then either $D\cap E=\emptyset$ or $E=\{b,c\}$, for some $c$. If $\bigcap\limits_{A\in c^{-1}(l)} A=\emptyset$ then there exists another set $F$ with $b\not\in F$. $F$ has to intersect both $D$ and $E$, thus $F=\{a,c\}$. Hence $|c^{-1}(l)|\leq 3$. 
\end{proof}\qed

Define, for $r\geq 0$
\[
p_{r}=|\{1\leq \lambda \leq r: |c^{-1}(\lambda)|\geq 4\mbox{ and } \bigcap\limits_{A\in c^{-1}(\lambda)} A\neq \emptyset \}|,
\]
\[
s_{r}=|\{i\in [n]: \bigcap\limits_{A\in c^{-1}(\lambda)} A =\{i\}\mbox{ for some }1\leq \lambda \leq r \mbox{ with }|c^{-1}(\lambda)|\geq 4\}|,
\]
 (call such an $i$ counted by $s_{r}$ {\em special}) 
\[
M_{r}=\sum_{i=1}^{r} |c^{-1}(i)|,\mbox{ }
N_{r}=p_{r}(n-1) - \frac{p_{r}(p_{r}-1)}{2}+ 3(r-p_{r})
\]

\begin{lemma}
Sequences $M_{r},N_{r}$ are monotonically 
increasing.
\end{lemma} 
\begin{proof} 

First $p_{r+1}-p_{r}\in \{0,1\}$. Next $M_{r+1}-M_{r}=|c^{-1}(r+1)|\geq 0$. Finally, $N_{r+1}-N_{r}=3$ if $p_{r+1}=p_{r}$,  $N_{r+1}-N_{r}=(n-1)-p_{r}$ if $p_{r+1}-p_{r}=1$. In this latter case $p_{r}= p_{r+1}-1\leq (n-3)-1 = n-4$ hence $N_{r+1}-N_{r}\geq 3.$ 
\end{proof}\qed

We now prove  the following result: 

\begin{lemma}
For $1\leq r\leq n-3$, $M_{r} \leq  N_{r}$. 
\end{lemma} 
\begin{proof}

First $
s_{r}(n-1)-\frac{s_{r}(s_{r}-1)}{2}\leq p_{r}(n-1)-\frac{p_{r}(p_{r}-1)}{2}$. 
Indeed, the left hand side is 
\begin{flalign*}
& s_{r}(n-1)-(0+1+\ldots s_{r}-1)= \\ 
& = (n-1)+(n-1-1)+(n-1-2)+\ldots + (n-1-s_{r}+1)\\ 
& = (n-1)+(n-2)+\ldots + (n-s_{r})
\end{flalign*}
and similarly for the right-hand side. The desired inequality follows from the fact that $s_{r}\leq p_{r}$, valid since a special $i$ may be counted for two different $\lambda$. 

We prove the lemma by showing the stronger inequality
\begin{equation}
M_{r}\leq s_{r}(n-1) - \frac{s_{r}(s_{r}-1)}{2}+ 3(r-p_{r})
\label{stronger}
\end{equation}
 
The first two terms of the right-hand side of~(\ref{stronger}) count sets $\{p,q\}\in {{n}\choose {2}}$ with at least one special element. Indeed $s_{r}(n-1)$ is the number of {\em pairs} $(i,j)$ with $i\neq j$ and $i$ special. This formula overcounts {\em sets} with at least one special element when $j$ is special too (and set $\{i,j\}$ is counted for both pairs $(i,j)$ and $(j,i)$). The number of such pairs is 
precisely $\frac{s_{r}(s_{r}-1)}{2}$. 

Now $M_{r}$ sums up cardinalities of color classes $1$ to $r$. For those $\lambda$'s in $[r]$ such that $|c^{-1}(\lambda)|\geq 4$ and all sets in the color class intersects at a special $i$, all these sets contain a special value, hence they are also counted by the right-hand side of~(\ref{stronger}). The difference is made by the remaining $\lambda$'s (there are $r-p_{r}$ of them). By Claim~\ref{claim1} they add at most $3(r-p_{r})$ sets to $M_{r}$, establishing the desired result. 
\end{proof}\qed

\begin{lemma} We have $N_{n-3}\leq {{n}\choose {2}}-3$. 
\end{lemma} 
\begin{proof}
$N_{n-3}=(n-1)+(n-2)+\ldots + (n-p_{n-3})+ 3(n-3-p_{n-3})$. But $3(n-3-p_{n-3})\leq 3+4+\ldots + (n-p_{n-3}-1)$ hence 
\[
N_{n-3}\leq 3+4+\ldots + (n-1)=n(n-1)/2 - 1 - 2 = {{n}\choose {2}} - 3. 
\]
\end{proof}\qed

Now Theorem~(\ref{ub}) (a) follows by setting $r=n-3$. The right-hand side is 
${{n}\choose {2}}-3$. But there are $M_{n-3}={{n}\choose {2}}$ sets to cover. 
 \end{proof}\qed

\subsubsection*{Propositional simulation.}

Now we start translating the above proof into sequent calculus LK. We will sketch the nontrivial steps of the translation. Tedious but straightforward computations shows that all these steps amount to polynomial length proofs.

Lemma~\ref{claim1} can, for instance, be polynomially simulated as follows: 

\begin{lemma}For $n\geq 5$ and $1\leq l\leq n-3$ define the propositional formula \\ $Int_{n,l}[(X_{S,l})_{S\in {{n}\choose{2}}}]$ to be 
\[
\bigvee_{\stackrel{D,E\in {{n}\choose {2}}}{D\cap E=\emptyset}} (X_{D,l}\AND X_{E,l}) \OR [Count[(X_{S,l})]\leq 3]\OR \bigvee_{i\in [n]} \big(\bigwedge_{i\not \in S}\overline{X_{S,l}}\big).
\]
Here $Count$ are Buss's counting formulas. Then for every $1\leq l \leq n-3$ formula $
Ant_{n,2} \vdash Int_{n,l}$ 
has proofs of polynomial length in sequent calculus LK. 
\label{deriv:2}
\end{lemma}
\begin{proof}

We will apply the following trivial 

\begin{lemma} 
Let $A,B,C,D$ be four distinct subsets of cardinality 2 of [n]. Then at least one of the following alternatives holds: 
\begin{itemize}
\item At least two sets among $A,B,C,D$ are disjoint. 
\item $|A\cup B\cup C\cup D|=5$ and $|A\cap B\cap C\cap D|= 1.$
\end{itemize}
\end{lemma} 

The lemma will be used "at the meta level", that is it will not be codified propositionally, but simply used to argue for the correctness of the proof. 

Define (only for notational convenience, not as part of the Frege proof) shorthand 
\[
Z_{A,B,C,D}^{\mbox{ }l}:=X_{A,l}\AND X_{B,l}\AND X_{C,l}\AND X_{D,l}
\]

Now for any $1\leq l\leq n-3$ 
\[
Ant_{2,n},\neg [Count[(X_{S,l})_{S\in {{n}\choose{2}}}]\leq 3]\vdash \bigvee_{\stackrel{A,\ldots,D}{\mbox{\small distinct}}} (Z_{A,B,C,D}^{\mbox{ }l})
\]
On the other hand, when $|A\cap B\cap C\cap D|=\emptyset$ two of these sets must be disjoint, 
\[
\mbox{hence for such sets }Z^{l}_{A,B,C,D} \vdash \bigvee_{\stackrel{E,F\in \{A,\ldots,D\}}{E\cap F=\emptyset}} (X_{E,l}\AND X_{F,l})
\]
As for any $n\geq 5$ any two disjoint sets in ${{n}\choose {2}}$ are part of a 4-tuple of sets in 
${{n}\choose {2}}$
\[
\bigvee_{\stackrel{E,F\in \{A,B,C,D\}}{E\cap F=\emptyset}} (X_{E,l}\AND X_{F,l})\vdash \bigvee_{\stackrel{E,F\in {{n}\choose {2}}}{E\cap F=\emptyset}} (X_{E,l}\AND X_{F,l}), \mbox{ hence }
\]
\begin{equation}
Ant_{n,2},\neg [Count(X_{A,l})\leq 3]\vdash  \bigvee_{\stackrel{E,F\in \{A \ldots D\}}{E\cap F=\emptyset}} (X_{E,l}\AND X_{F,l})\mbox{  } \OR \bigvee_{\stackrel{A,B,C,D\subseteq [n]}{|A\cap B\cap C\cap D|=1}} Z^{\mbox{ }l}_{A,B,C,D}
\label{foo}
\end{equation}
Now we rewrite 
\[
\bigvee_{\stackrel{A,B,C,D\subseteq [n]}{|A\cap B\cap C\cap D|=1}}Z^{l}_{A,B,C,D}  = \bigvee_{i\in [n]} \big(\bigvee_{A\cap B\cap C\cap D=\{i\}} Z^{\mbox{ }l}_{A,B,C,D}\big)
\]
Fix an arbitrary 4-tuple $(A,B,C,D),$ $A\cap B\cap C\cap D=\{i\}$. For any $H\in {{n}\choose {2}}$, $H\not \ni i$ one of the sets $A,B,C,D$ is disjoint from $H$. Hence by modus ponens (cut) with $E=H$ and $F\in \{A,B,C,D\}$ with $H\cap F = \emptyset$
\[
Ant_{2,n},Z^{\mbox{ }l}_{A,B,C,D},\bigwedge_{\stackrel{E,F\in {{n}\choose {2}}}{E\cap F=\emptyset}}(\overline{X_{E,l}}\OR\overline{X_{F,l}})\vdash \overline{X_{H,l}}
\]
By repeatedly introducing ANDs in the conclusion, then OR in the antecedent  
\[
Ant_{2.n},\bigvee_{\stackrel{A,B,C,D\subseteq [n]}{A\cap B\cap C\cap D=\{i\}}} Z^{\mbox{ }l}_{A,B,C,D},\bigwedge_{\stackrel{E,F\in {{n}\choose {2}}}{E\cap F=\emptyset}}(\overline{X_{E,l}}\OR\overline{X_{F,l}})\vdash \bigwedge_{\stackrel{H\in {{n}\choose {2}}}{H\not \ni i}}\overline{X_{H,l}}
\]
By repeated introduction of ORs in both the antecedent and  the conclusion 
\[
Ant_{2,n},\bigvee_{i\in [n]}\big(\bigvee_{\stackrel{A,B,C,D\subseteq [n]}{A\cap B\cap C\cap D=\{i\}}} Z^{\mbox{ }l}_{A,B,C,D} \big) ,\bigwedge_{\stackrel{E,F\in {{n}\choose {2}}}{E\cap F=\emptyset}}(\overline{X_{E,l}}\OR\overline{X_{F,l}})\vdash \bigvee_{i\in [n]}\big(\bigwedge_{\stackrel{H\in {{n}\choose {2}}}{H\not \ni i}}\overline{X_{H,l}}\big)
\]
Taking into account~(\ref{foo}) and moving the third antecedent on the right-hand side we get the proof of Lemma~\ref{deriv:2}. 
\end{proof}\qed

\begin{definition}
Define for $i\in [n]$, $l \in [n-3]$ formula  
\begin{flalign*}
Special_{i,l}[(X_{S,l})_{S\in {{n}\choose {2}}}]\equiv \big[\big(Count[(X_{S,l})_{S\in {{n}\choose {2}}}]\geq 4\big)\big] \AND  \big[\big(\bigwedge_{\stackrel{B\in {{n}\choose{2}}}{B\not \ni i}} \overline{X_{B,l}}\big)\big]
\end{flalign*}
For $r\in [n-3]$ let 
$q_{r}$  be the number of indices $i\in [n]$ such that there is a color $l$, $1\leq l \leq r$ with  $Special_{i,l}[(X_{S,l})_{S\in {{n}\choose {2}}}]=TRUE$. 
\end{definition} 

\begin{remark}
Semantically we have $q_{r}=s_{r}$ (in $q_{r}$ we do not require that the intersection of all sets $B\in {{n}\choose{2}}\cap c^{-1}(l)$ have cardinality {\em exactly one}, but that is true if $Count[(X_{S,l})_{S\in {{n}\choose {2}}}]\geq 4$)
\end{remark}

Given $r\leq n-3$ we can compute, using a Frege proof, the binary representation of $q_{r}$.  
as $
q_{r} = Count( \{i\in [n]\mbox{ } |\mbox{ } \bigvee\limits_{l=1}^{r} Special_{i,l}\}).$
Now define for $0\leq r\leq n-3$
\begin{align*}
& M_{r}=|\{A\in {{n}\choose {2}}\mbox{:} \bigvee_{1\leq l\leq r}X_{A,l} \}|\mbox{ } \big(\mbox{ semantically }=\sum_{i=1}^{r} |c^{-1}(i)|\big)\\
& M_{r}^{(1)}=|\{A\in {{n}\choose {2}}:\bigvee_{1\leq l\leq r} (X_{A,l}\AND [Count(X_{S,l})\leq 3])\mbox{ }\}|
\\
& M_{r}^{(2)}= |\{A\in {{n}\choose {2}}:\bigvee_{1\leq l\leq r} (X_{A,l}\AND [Count((X_{S,l})_{S\in {{n}\choose {2}}})\geq 4])\mbox{ }\}| \\
& Q_{r}^{(1)}=|\{l\mbox{ }|\mbox{ }(1\leq l\leq r)\AND  [Count(X_{S,l})\leq 3]\mbox{ }\}|,
\end{align*}

One can easily prove in LK  the following 
\begin{lemma} $
Ant_{2,n}\AND Onto_{2,n}\vdash \big[M_{r}=M_{r}^{(1)}+M_{r}^{(2)}\big].
$
\label{clm:mr}
\end{lemma} 

\begin{lemma}
One can compute in LK the binary expansions of $M_{r}^{(1)}$, $Q_{r}^{(1)}$ and prove that 
$
Ant_{2,n}\AND Onto_{2,n}\vdash \big[M_{r}^{(1)}\leq 3\cdot Q_{r}^{(1)}\big]
.$
\end{lemma} 
\begin{proof}
For the first part we use Buss's counting approach. For the second, define 
\[
W_{l} = \left\{ \begin{array}{ll}
1 & \mbox{if $Count(X_{S,l})\leq 3,$}\\
0 & \mbox{ otherwise}.\end{array} \right. \mbox{ 

and }
Y_{l} = \left\{ \begin{array}{ll}
Count(X_{S,l}) & \mbox{if $Count(X_{S,l})\leq 3,$}\\
0 & \mbox{ otherwise}.\end{array} \right.
\] 
Then (one can readily prove in LK that) $Y_{l}\leq 3W_{l}$. Summing up we get $M_{r}^{(1)}\leq 3Q_{r}^{(1)}$. The proof (using the fact that the cardinal of a union of disjoint sets is the sum of cardinals of individual subsets) can easily be simulated in LK. 
\end{proof}\qed

\begin{definition} Let 
\[
P_{r}^{(2)}=|\{A\in {{n}\choose {2}}: \bigvee_{1\leq l\leq r} X_{A,l}\AND \big[\big(\bigwedge_{B\not \ni First(A)} \overline{X_{B,l}}\big) \oplus \big(\bigwedge_{B\not \ni Second(A)} \overline{X_{B,l}}\big)\big] \}|.
\]
where 
\begin{itemize}
\item $First(A)$ is the smallest element in $A$, $Second(A)$ is the largest. 
\item $P\oplus Q$ in the above expression is a shorthand for $(P\AND \overline{Q})\OR(Q\AND \overline{P})$.  Since there are $O(n^2)$ sets $B$ to consider, the size of the formula after expanding to CNF is $O(n^4)$.
\end{itemize}
\end{definition}

\begin{lemma} One can prove in LK that 
\begin{flalign*}
Ant_{2,n}\AND Onto_{2,n}\AND \neg Cons_{2,n}\vdash [M_{r}^{(2)}& \leq  P_{r}^{(2)}].
\end{flalign*}
\end{lemma} 

\begin{proof} 
The inequality follows in the following way: From Lemma~\ref{deriv:2}
\begin{align*}
& Ant_{2,n}\vdash Int_{n,l}\mbox{, hence }\\
& Ant_{2,n} \AND \neg Cons_{2,n} \AND X_{A,l}\AND [Count((X_{S,l})_{S\in {{n}\choose {2}}})\geq 4]\vdash  \bigvee_{i\in [n]} Special_{i,l}
\end{align*}

Now assume $X_{A,l}\AND [Count((X_{S,l})_{S\in {{n}\choose {2}}})\geq 4])$. 
For $i\neq First(A), Second(A)$ set $A$ is among the $B$'s in the conjunction defining $Special_{i,l}$, so all these  formulas evaluate to FALSE. 
Furthermore, if $X_{A,l}$ and $Count ((X_{S,l})_{S\in {{n}\choose {2}}})\geq 4]$ then exactly one of the two remaining terms, $\bigwedge\limits_{B\not \ni First(A)} \overline{X_{B,l}}$ and $\bigwedge\limits_{B\not \ni Second(A)} \overline{X_{B,l}}$ also simplifies to FALSE. 
Indeed, there is a set $B\neq A$ with $X_{B,l}$. $B$ does not contain one of $First(A),Second(A)$, hence $\overline{X_{B,l}}$ appears in exactly one of the corresponding conjunctions, 
making it FALSE.  

Hence every set $A$ counted by $M_{r}^{(2)}$ is among those counted by $P_{r}^{(2)}$ and, by $Onto_{2,n}$, only in one such set. 

%We finally infer 
%\[
%Ant_{n,2} \AND X_{A,l}\AND [Count((X_{S,l})_{S\in {{n}\choose {2}}})\geq 4])\vdash X_{A,l}\AND \bigvee_{i\in [n]} Special_{i,l}
%\]

\end{proof}\qed

%Now define

Define $U_{r} =  |\{A\in {{n}\choose {2}}: \big(\bigvee\limits_{\lambda=1}^{r} Special_{First(A),\lambda}\big) \AND \big(\bigvee\limits_{\nu=1}^{r} Special_{Second(A),\nu} \big) \}|.$
%\end{flalign*}

\begin{lemma}
We have (and can prove in polynomial size in LK)
\begin{flalign*}
& Ant_{n,2}\vdash  \big[U_{r}= |\{(i,j): i<j\in [n]\mbox{ and } \big(\bigvee\limits_{\lambda=1}^{r} Special_{i,\lambda}\big)\AND  \\ 
& \AND \big( \bigvee\limits_{\nu=1}^{r}Special_{j,\nu}\big)\}|={{q_{r}}\choose {2}}.\big] 
\end{flalign*}
\end{lemma}
\begin{proof}
The first equality amounts to no more than semantic reinterpretation. The last equality follows from Lemma~\ref{frege-count} (2).
\end{proof}\qed

\begin{lemma} $
Ant_{2,n}\AND Onto_{2,n}\vdash \big[ U_{r}+P_{r}^{(2)}\leq q_{r}\cdot (n-1)\big]$ has poly-size LK proofs. 
\end{lemma}
\begin{proof} 
%First, the left-hand side of the inequality in the consequent  is at most 
%\[
%|\{(i,j): i<j\in [n]\mbox{ and } \bigvee_{1\leq \lambda\leq r} Special_{i,\lambda} \OR \bigvee_{1\leq \nu\leq r} 5Special_{j,\nu}\}|
%\]

$U_{r}$ counts sets $\{i,j\}$ such that {\em both} $i$ and $j$ are special. $P_{r}^{(2)}$ counts sets $A$ for which exactly one of $First(A), Second(A)$ is special for the unique $l$ (by $Ant_{2,n}\AND Onto_{2,n}$) such that $X_{A,l}$. Therefore 
\[
P_{r}^{2}\leq |\{(i,j): i<j\in [n]\mbox{ and } \big[ \bigvee_{1\leq \lambda\leq r} Special_{i,\lambda}\big] \oplus \big[\bigvee_{1\leq \nu\leq r} Special_{j,\nu}\big]\}|
\]

Let $X_{i,j}=1$ if $i<j\in [n]$ and both $i$ and $j$ are special or if $i<j\in [n]$ and exactly one of $i,j$ is special. Then $U_{r}+P_{r}^{2}\leq Count[(X_{i,j})]$. 

The right-hand side is, by Lemma~\ref{frege-count} (4), equal to 
\[
|\{ (i,j): i\neq j\in[n]\mbox{ and } \bigvee_{1\leq \lambda\leq r} Special_{i,\lambda} \}|=Count[(Y_{i,j})]
\]
where $Y_{i,j}= [\bigvee_{1\leq \lambda\leq r} Special_{i,\lambda}] \cdot \delta_{\{i\neq j\}}$. Now 
 $X_{i,j}\leq Y_{i,j}$ if both $i,j$ are special or $i$ is but $j$ isn't,  
$X_{i,j}\leq Y_{j,i}$ if $j$ is special but $i$ isn't, and we apply Lemma~\ref{frege-count} (5). \end{proof}\qed

\begin{corollary} 
LK can efficiently prove formulas: 
\begin{flalign*}
& (1).\mbox{ } Ant_{2,n}\AND Onto_{2,n}\AND \neg Cons_{2,n}\vdash [M_{r}+U_{r}\leq q_{r}(n-1)+3\cdot Q_{r}^{1}].\\
& (2).\mbox{ } Ant_{2,n}\AND Onto_{2,n}\vdash [M_{n-3}={{n}\choose {2}}].\\
& (3).\mbox{ } Ant_{2,n}\AND Onto_{2,n}\vdash [q_{n-3}\leq n-3]\\
& (4).\mbox{ } Ant_{2,n}\AND Onto_{2,n}\vdash [q_{n-3}(n-1)+3\cdot Q_{n-3}^{1}+{{n-3}\choose {2}}\leq \\ 
& \leq (n-3)\cdot (n-1)+U_{n-3}]\\
\end{flalign*}
\end{corollary}

\begin{proof} 
The conclusions can be derived from the antecedent $Ant_{2,n}\AND Onto_{2,n}\AND \neg Cons_{2,n}$ by simulating the following arguments: 
\begin{flalign*} 
& (1).\mbox{ } M_{r}+U_{r}=(M_{r}^{1}+M_{r}^{2})+U_{r}\leq 3Q_{r}^{(1)}+P_{r}^{(2)}+U_{r}\leq q_{r}(n-1)+3Q_{r}^{(1)}. 
\\ 
& (2).\mbox{ } \mbox{follows by applying Lemma~\ref{frege-count}(1)}\\
& (3).\mbox{ } \mbox{follows from the formula defining }q_{r}\\
& (4).\mbox{ } 3\cdot Q_{n-3}^{1}=3\cdot (n-3-p_{n-3})\leq 3\cdot (n-3-q_{n-3}) \mbox{ as }q_{n-3}\leq p_{n-3}\\
& \mbox{
(since there may be more than one color class sharing the same special element).}\\
&\mbox{ Now } R.H.S. - L.H.S. \geq (n-3)(n-1)-q_{n-3}(n-1)-3(n-3-q_{n-3})-\\ & - \frac{(n-3)^2-(n-3)}{2}+
\frac{q_{n-3}^2-q_{n-3}}{2} \geq (n-3-q_{n-3})(n-4)-\\ & - \frac{(n-3-q_{n-3})(n-3+q_{n-3}-1)}{2}
 = \frac{(n-3-q_{n-3})(n-4+q_{n-3})}{2}\geq 0.
\end{flalign*} 
\end{proof} \qed 
Now we can put everything together to prove Theorem~\ref{ub} (a): by (1) 
\[
Ant_{2,n}\AND Onto_{2,n}\AND \neg Cons_{2,n} \vdash [M_{n-3}+U_{n-3}\leq q_{n-3}(n-1)+3\cdot Q_{n-3}^{1}].
\]
Adding relation (4), taking into account (2) and simplifying by $U_{n-3}+3Q_{n-3}^{1}$ we get 
\[
Ant_{n,2}\AND Onto_{n,2}\AND \neg Cons_{2,n}\vdash {{n}\choose {2}}+{{n-3}\choose {2}}\leq (n-1)(n-3),
\]
or, equivalently 
\[
Ant_{n,2}\AND Onto_{n,2}\AND \neg Cons_{2,n}\vdash [2n^2-8n+12\leq 2n^2-8n+6]\vdash \square
\]
Moving $\neg Cons_{2,n}$ to the other side we get the desired result. 
\qed

\subsection{\bf Case $k=3:$ }

A claim similar to Lemma~\ref{claim1} holds for $k=3$ (for a proof that can be efficiently simulated in EF (in fact Frege) see the Appendix): 

\begin{lemma}\cite{garey-johnson-coloring}
For any $1\leq \lambda \leq n-5$ at least one of the following is true: 
\begin{itemize} 
\item $c^{-1}(\lambda)$ contains two disjoint sets
\item $|c^{-1}(\lambda)|\leq 3n-8$, or 
\item there exists $x\in \cap_{A\in c^{-1}(\lambda)} A$. 
\end{itemize} 
\label{claim:k3} 
\end{lemma}  

Assuming this claim we settle the case $k=3$. The argument we give is simpler than the argument in \cite{garey-johnson-coloring}, and has the advantage of being easily/efficiently simulated in  EF, similar to the case of PHP. Full details are deferred to the journal version.  

\begin{lemma}\cite{garey-johnson-coloring}
Kneser's conjecture is true for $k=3$. 
\end{lemma}
\begin{proof} 
By induction. The base case $n=7$ can be verified directly. 
Assume that one could give a coloring $c$ of the Kneser graph $KG_{n,3}$ with $n-5$ colors. If there is a color $\lambda$ with 
$x\in \cap_{A\in c^{-1}(l)} A$ then one could eliminate both element $x$ and color $\lambda$, obtaining a $n-6$ coloring of graph $KG_{n-1,3}$, thus contradicting the inductive hypothesis. 

If no color class contains two disjoint sets then all of them satisfy $|c^{-1}(l)|\leq 3n-8$. But then we 
would have $
{{n}\choose {3}}\leq (n-5)(3n-8)$. This is false for $n\geq 7$. 
\end{proof}\qed

We could try to give a Frege proof of the case $k=3$ based on counting principles, using the following strategy, similar to the one used in case $k=2$: 
\begin{enumerate} 
\item Count, using a Frege proof, the number $p_{r}$ of sets $c^{-1}(l)$, $1\leq l \leq t$ such that $|c^{-1}(\lambda)|\geq 3n-7$  [implicitly $\cap_{A\in c^{-1}(l)} A\neq \emptyset$]
\item  Define $
M_{r}^{(3)} = \sum_{i=1}^{r} |c^{-1}(i)|$ and 
\[
N_{r}^{(3)}={{n-1}\choose {2}}+{{n-2}\choose {2}}+\ldots + {{n-p_{r}}\choose {2}}+ (n-5-p_{r})(3n-7)
\]
\item Show inductively that $M_{r}^{(3)} \leq  N_{r}^{(3)}$.  
\item Obtain a contradiction from $M_{n-5}^{(3)}={{n}\choose {3}}$ and $N_{n-3}^{(3)}< {{n}\choose {3}}$. 
\end{enumerate} 

Although some of this program can be carried through, {\bf this approach does not seem to work.}  The inequality that critically fails is the last one: 
when $k=2$  we showed that $N_{n-3}^{(2)}< {{n}\choose {2}}$ as the maximum of the upper bound was obtained for $p_{r}=n-3$. For $k=3$, though, such a statement is not true. Indeed, since 

\begin{enumerate} 
\item we need to give upper bound estimates on the size of $n-5$ color classes. 
\item $3n-8$, the bound on the size of independent sets is growing with $n$
\end{enumerate} 

\noindent one cannot guarantee that $N_{n-3}^{(3)}< {{n}\choose {3}}$ for all possible values of $p_{r}$. For instance, if  $p_{r}=n-6$ (an event we cannot exclude) the resulting upper bound, 
${{n-1}\choose {2}}+{{n-2}\choose {2}}+\ldots + {{6}\choose {2}}+(3n-8)={{n}\choose {3}}-10-6-3-1+(3n-7)={{n}\choose {3}}+(3n-27)$ 
is not smaller than ${{n}\choose {3}}$ for $n\geq 10$. For this reason when $k=3$ we will have to do with the {\em extended Frege} proof described above. 

Lemma~\ref{claim:k3} can be efficiently simulated in EF (actually in Frege) via a straightforward but tedious adaptation of the argument in \cite{garey-johnson-coloring} (see the Appendix).  

On the other hand it may still be possible (and we conjecture that this can be done) to obtain a Frege proof by a more refined version of the above counting approach: rather than just counting "large" color classes (those with cardinality at least $3n-7$) we could try to make a finer distinction (based on the structure of color classes 
displayed by the proof of Lemma~\ref{claim:k3}) to obtain tighter upper bounds for $M_{r}^{(3)}$.

%\endproof
\qed

\section{Heads up: the general case of the Kneser-Lov\'{a}sz Theorem}

In this section we briefly announce the other results on the proof complexity of the Kneser-Lov\'{a}sz Theorem presented in a companion paper \cite{istrate-kneser-2}. Unlike the cases $k=2,3$, cases $k\geq 4$ apparently require proof systems more powerful than EF. Indeed, the general case of the Kneser-Lov\'{a}sz theorem follows by a combinatorial result known as the {\em octahedral Tucker lemma} \cite{matousek-book}. The propositional counterpart of this implication is the existence of a variable substitution that transforms the propositional encoding of the octahedral Tucker lemma into the Kneser-Lov\'{a}sz formulae. 

Though the formalization of the octahedral Tucker lemma  yields a formula of exponential size, the octahedral Tucker lemma admits \cite{istrate-kneser-2} a (nonstandard) version leading to polynomial-size formulas that is sufficient to prove the Kneser-Lov\'{a}sz theorem. However, even this version seems to require exponentially long EF proofs. The reason is that we prove the Octahedral Tucker Lemma by reduction to a Tseitin formula, crucially, though, to one on a complete graph $K_{m}$ of {\bf exponential size} ($m=O(n!\cdot 2^{n})$). 

The (exponentially long) proofs of these exponential Tseitin formulas can be generated implicitly \cite{krajivcek2004implicit}. However, not only the proof steps but the very {\em formulas} involved in the proof may have exponential size and need to be generated implicitly. Implicit proofs with implicitly generated formulas have been previously considered in the literature \cite{krajicek2004diagonalization}. We postpone the discussion of further technical details to \cite{istrate-kneser-2}.

\section{Conclusions, open problems and acknowledgments}

Our work has introduced a new class of propositional formulas to investigate with respect to complexity, and raises several open questions:  
\begin{enumerate} 
\item Does $Kneser_{2,n}$ have polynomial size cutting plane proofs/OBDD with projection,  as PHP does \cite{cook1987complexity,chen2009direct}? 
 \item Does family $Kneser_{3,n}$ have polynomial size Frege proofs  ? 
\item Is  family $Kneser_{k,n}$, for $k\geq 4$, hard for Frege/EF proofs ? 
\item There is a large and reasonably sophisticated literature dealing with extensions of the Kneser-Lov\'{a}sz Theorem (see e.g. \cite{kozlov-book}) or other results in combinatorial topology \cite{matousek-book,de2012course}.  Further 
investigate such results from the standpoint of bounded reverse mathematics. 
\end{enumerate}

 This work has been supported by CNCS IDEI Grant PN-II-ID-PCE-2011-3-0981 "Structure and computational difficulty in combinatorial optimization: an interdisciplinary approach".

\bibliographystyle{alpha}
\bibliography{/home/gistrate/Dropbox/texmf/bibtex/bib/bibtheory}

\section*{Appendix}

\subsection{(Extended) Frege proof of Claim~\ref{claim:k3}}

\begin{proof}

The following (semantical) argument is just a rewriting of the original proof of Lemma 1 from the Appendix of \cite{garey-johnson-coloring}. It is included in detail to make the paper self-contained and support the claim that this argument could be simulated by Frege proofs. 

Assume that $c^{-1}(\lambda)\neq \emptyset$. Let $\{a,b,c\}\in c^{-1}(\lambda)$. Define: 
\[
A=\{W \in c^{-1}(\lambda): a\in W, b\not \in W\}, 
B=\{W \in c^{-1}(\lambda): a\not \in W, b \in W\}
\]
\[
C=\{W\in c^{-1}(\lambda): c\in W, a \not \in W, b \not\in W\}, 
D=\{W \in c^{-1}(\lambda): a \in W, b \in W\}, 
\]
\begin{lemma}
$c^{-1}(\lambda)$ contains two disjoint sets, or families $A,B,C,D$ partition $c^{-1}(\lambda)$.  
\end{lemma} 
\begin{proof} 
Disjointness is easy. The partitioning follows since $(a,b,c)\in c^{-1}(\lambda)$, hence every set in $c^{-1}(\lambda)$ must contain one of $a,b,c$. 
\end{proof} \qed
\begin{corollary}
$c^{-1}(\lambda)$ contains  disjoint sets or $|c^{-1}(\lambda)|=|A|+|B|+|C|+|D|$. 
\end{corollary} 

\begin{lemma} Assume $\{a,b,c\}$ is chosen so that $|A|\geq |B|\geq |C|$. Then at least one of the following alternatives holds: 

\begin{enumerate} 
\item $c^{-1}(\lambda)$ contains  disjoint sets,
\item  $\bigcap\limits_{W\in c^{-1}(\lambda)} W\neq \emptyset$, or 
\item $B\neq \emptyset$ and $|A|\leq (n-3)$ and $|A|+|B|\leq 2n - 6 $.
\end{enumerate} 

\label{ab}
\end{lemma} 
\begin{proof} 
A case  analysis: 
\begin{itemize} 

\item{\bf Case 1: $B= \emptyset$.} 

Then $C=\emptyset$ as well. Consequently $\bigcap\limits_{W\in c^{-1}(\lambda)} W\ni a.$

\item{\bf Case 2: there are sets $W_{1},W_{2}\in B$ with $W_{1}\cap W_{2}=\{b\}$ (implicitly $B\neq \emptyset$ ).} 

Then either $c^{-1}(\lambda)$ contains two disjoint sets or every set $W\in A$ must meet both $W_{1}$ and $W_{2}$ (in an element obviously different from $b$). There are at most 4 such sets $W$ (corresponding to the $2\time 2$ choices of elements from $W_{1},W_{2}$) hence $|B|\leq |A|\leq 4$ and $|A|+|B|\leq 8\leq 2n-6$ for $n\geq 7$. 
\item{\bf Case 3: $|B|=1$.}

Let $B=\{b,i,j\}$. Then either $c^{-1}(\lambda)$ contains two disjoint sets or every set $W\in A$ must contain either $i$ or $j$ but not $i$. There are at most $n-3$ sets of the first type and at most $n-4$ of the second, hence 
$|A|+|B|\leq (n-3)+(n-4)+1=2n-6$. 

\item{\bf Case 4: $|B|=2$ but for the two sets $W_{1},W_{2}\in B$ we have $W_{1}\cap W_{2}=\{b\}$.} 

Let $W_{1}=\{b,i,j\}, W_{2}=\{b,i,k\}$ with $i,j,k\neq a,b$. Then either $c^{-1}(\lambda)$ contains two disjoint sets or every set $W\in A$ must contain either $i$ or both $j$ and $k$. 

There are at most $n-3$ sets $\{a,i,l\}$, $l\neq a,b,i$ of the first type and one set, $\{a,j,k\}$, of the second. 
Hence $|A|+|B|\leq (n-3)+1+2=n\leq 2n-6$. 

\item{\bf Case 5: $|B|\geq 3$ and $|\bigcap\limits_{W\in B} W|=2$.} 

Let $\bigcap\limits_{W\in B} W =\{b,i\}$. Since $|B|\geq 3$ there exist distinct indices $j,k,l$ such that  
$\{b,i,j\},\{b,i,k\},\{b,i,l\} \in B$. 

If there is $W\in A$ that does not contain $i$ it follows that $W$ is disjoint from at least one of these. 

Otherwise all sets in $A$ contain $i$. There are at most $(n-3)$ such sets $\{a,i,r\}$, $r\neq i,b$. Hence 
$|A|+|B|\leq 2\cdot |A|\leq 2(n-3)= 2n-6$. 

\item{\bf Case 6: $|B|\geq 3$, $\bigcap\limits_{W\in B} W=\{b\}$ and for all $Z,T\in B$,  $|Z\cap T|\geq 2$.} 

Let $W_{1}=\{b,i,j\}$. Let $W_{2}\in B$, $i\not \in W_{2}$. $W_{2}$ exists by the second condition. By the third condition $j\in W_{2}$. By the same reason there exists $W_{3}\in B$, $i\in W_{3}, j\not \in W_{3}$. 

Let $W_{2}=\{b,j,k\}$, $W_{3}=\{b,i,l\}$. $k$ must be equal to $l$ so that $|W_{2}\cap W_{3}|\geq 2$.

By the third condition it follows that $B=\{W_{1},W_{2},W_{3}\}$. 

Now either $c^{-1}(\lambda)$ contains two disjoint sets or every set in $A$ must contain two of 
$i,j,k$. There are at most three such sets, so $|A|+|B|\leq 6\leq 2n-6$. 

\end{itemize} 
\end{proof}\qed 

\begin{lemma} Assume $\{a,b,c\}\in c^{-1}(\lambda)$ is chosen such that $|A|\geq |B|\geq |C|$. Then  [$c^{-1}(\lambda)$ contains  disjoint sets],  or [
$|C|+|D|\leq n - 2 $].
\label{cd}
\end{lemma} 
\begin{proof} 
A case  analysis: 
\begin{itemize} 
\item {\bf Case 1: $C=\emptyset$.}

Since clearly $|D|\leq n-2$, $|C|+|D|\leq n-2.$
\item{\bf Case 2: $|C|=1$.} 

Let $C=\{c,i,j\}$. Then either $c^{-1}(\lambda)$ contains  disjoint sets or every $W\in D$ must contain one of $c,i,j$. There are three such sets, hence $|C|+|D|\leq 3+1 = 4\leq n-2$. 

\item{\bf Case 3: $|C|\geq 2$ and $|\bigcap\limits_{W\in C} W|=2.$} 

Since $|C|\geq 2$ for any of the elements $j\not \in \bigcap\limits_{W\in C} W$ there exists a set $W\in C$ that does not contain it. 

Consider any set $Z=\{a,b,\lambda\}\in D$. If $\lambda = j$ then there exist two disjoint sets $W,Z \in c^{-1}(\lambda).$ The same conclusion is true if $\lambda \neq c,i$. In the opposite case we conclude that  $|D|\leq 2$. But $|C|\leq (n-4)$, since $a,b,c,i$ are forbidden options for any third member of a set in $C$. Thus $|C|+|D|\leq (n-2)$. 

\item{\bf Case 4: $|C|\geq 2$ and $\bigcap\limits_{W\in C} W=\{c\}.$} 

Let $W=\{a,b,i\}$ in $D$. If some $Z\in C$ does not contain $i$ then $W,Z\in c^{-1}(\lambda), W\cap Z=\emptyset$. 

In the opposite case every set $Z\in C$ must contain $i$. By the hypothesis it follows that $|D|\leq 1$. On the other hand $|C|\leq |B|\leq |A|$. By previous lemma $|C|\leq (n-3)$, hence $|C|+|D|\leq (n-2)$. 

\end{itemize} 
\end{proof}\qed

Note that, since all indices in the proofs above range on sets of polynomial cardinality ($[n], {{n}\choose {3}}$, etc.) we could simulate the arguments above  even with {\em Frege proofs} without significant issues, along the lines of the translation done in the case $k=2$. For instance, the cardinality of sets $A,B,C,D$ is encoded by applying formulas $Count_{n}$ to appropriately chosen sets of variables. For instance 
\[
|A|=Count[(X_{W,l})_{W\ni a,W\not \ni b}]
\]
Statements $|A|\geq |B|$ and $|B|\geq |C|$ can be encoded propositionally, and the above argument yields, for every 
$\{a,b,c\}$ a propositional proof of a statement of type $\Phi_{a,b,c}\vdash \Xi_{a,b,c}$, where 
$\Phi_{a,b,c}$ encodes the antecedent and Onto formulas, plus condition $|A|\geq |B|\geq |C|$, and $\Xi_{a,b,c}$ encodes the conclusion of Claim~\ref{claim:k3}. 

Alternate cases in the proofs of Lemmas~\ref{ab} and~\ref{cd} translate to disjunctions in the propositional formulations, the way (for $k=2$) the three alternatives in Lemma~\ref{claim1} translated to a disjunction in the propositional formula $Int_{n,l}$ in Lemma~\ref{deriv:2}. We omit further details. 

Now all we need to prove the desired result, by combining the previous two lemmas, is that if $c^{-1}(\lambda)\neq \emptyset$ then for some $\{a,b,c\}\in c^{-1}(\lambda)$ it holds  that $|A|\geq |B|\geq |C|.$

This only needs to be argued at the semantic level: the propositional translation of the conditional argument given the "good set" $\{a,b,c\}$ is then enough to give the proof of the desired result. 

To this end choose, as specified in \cite{garey-johnson-coloring}
\begin{itemize} 
\item $a$ so that it maximizes $|\{W\in c^{-1}(\lambda): a\in W\}|,$
\item $b$ among sets $\{a,i,j\}\in c^{-1}(\lambda)$ so that it maximizes $|\{W\in c^{-1}(\lambda):b\in W, a\not \in W\}|,$ 
\item $c$ among sets $\{a,b,l\}\in c^{-1}(\lambda)$ to maximize $|\{W\in c^{-1}(\lambda):c\in W, a,b\not\in W \}|$ 
\end{itemize} 
We have 
\[
|\{W\in c^{-1}(\lambda): a\in W\}|\geq |\{W\in c^{-1}(\lambda): b\in W\}|
\]
hence 
\begin{align*}
|A|& =|\{W\in c^{-1}(\lambda): a\in W\}|-|\{W\in c^{-1}(\lambda): a,b\in W\}|\geq \\ 
& \geq |\{W\in c^{-1}(\lambda): b\in W\}|-|\{W\in c^{-1}(\lambda): a,b\in W\}|= |B|.
\end{align*}
Similarly 
\begin{align*}
& |B|=  |\{W\in c^{-1}(\lambda): b\in W, a\not \in W\}| \\
& \geq  |\{W\in c^{-1}(\lambda): c\in W, a\not \in W\}|\geq \\
& \geq |\{W\in c^{-1}(\lambda): c\in W, a\not \in W, b\not \in W\}|= |C|.
\end{align*}
\end{proof}\qed
\end{document}